\documentclass[usenames,dvipsnames,11pt]{article}

\usepackage{graphicx}
\usepackage{amsmath}
\usepackage{amssymb}
\usepackage{latexsym}
\usepackage{mathrsfs}
\usepackage{amsthm}
\usepackage{setspace}
\usepackage{epsfig}
\usepackage{subfigure}
\usepackage{authblk}
\usepackage{amsfonts}
\usepackage{url}
\usepackage{float}
\usepackage{multibib}

\usepackage[authoryear,round]{natbib}	
\usepackage{mathtools}														 
\mathtoolsset{showonlyrefs=false}									 

\usepackage{geometry}
\usepackage[dvips]{color}
\newtheoremstyle{newplain}
{4pt}
{4pt}
{\itshape}
{}
{\itshape\bf}
{.}
{.5em}
{}
\theoremstyle{newplain}
\newtheorem{theorem}{Theorem}[section]
\newtheorem{proposition}{Proposition}[section]

\newtheorem{lemma}{Lemma}[section]

\newtheorem{assumption}{Assumption}[section]

\newtheoremstyle{newdefinition}
{4pt}
{4pt}
{}
{}
{\itshape\bf}
{.}
{.5em}
{}
\theoremstyle{newdefinition}

\newtheorem{vg}{Example}[section]

\newtheorem{remark}{Remark}[section]
\numberwithin{equation}{section}
\newcommand{\be}{\begin{equation}}
\newcommand{\ee}{\end{equation}}
\newcommand{\nee}{\nonumber\end{equation}}
\newcommand{\eel}[1]{\label{#1}\end{equation}}
\newcommand{\brmk}[1]{\begin{remark}\label{#1}\begin{em} }
\newcommand{\ermk}{ $\quad\triangleleft$\end{em}\end{remark}}
\newcommand{\bvg}[1]{\begin{vg}\label{#1}\begin{em} }
\newcommand{\evg}{ $\quad\triangleleft$\end{em}\end{vg}}

\usepackage{chngcntr}
\counterwithout{equation}{section}

\begin{document}
\bibliographystyle{plainnat}
\setlength{\abovedisplayskip}{8pt}
\setlength{\belowdisplayskip}{8pt}
\setlength{\abovedisplayshortskip}{4pt}
\setlength{\belowdisplayshortskip}{8pt}

\begin{titlepage}

\title{{\bf{\Huge The Long Bond, Long Forward Measure and Long-Term Factorization in Heath-Jarrow-Morton Models}}}
\author{Likuan Qin\thanks{likuanqin2012@u.northwestern.edu}
}
\author{Vadim Linetsky\thanks{linetsky@iems.northwestern.edu}
}

\affil{\emph{Department of Industrial Engineering and Management Sciences}\\
\emph{McCormick School of Engineering and Applied Sciences}\\
\emph{Northwestern University}}
\date{}
\end{titlepage}
\maketitle

\begin{abstract}
This paper proves existence of the long bond, long forward measure and long-term factorization of the stochastic discount factor (SDF) of \citet{alvarez_2005using} and \citet{hansen_2009} in Heath-Jarrow-Morton (HJM) models in the function space framework of \citet{filipovic_2001consistency}. A sufficient condition on the weight in the Hilbert space of forward rate volatility curves is given that ensures existence of the long bond volatility process, the long bond process and the long-term factorization of the SDF into discounting at the rate of return on the long bond and a martingale component defining the long forward measure, the long-term limit of $T$-forward measures.

\end{abstract}

\section{Introduction}

The stochastic discount factor (SDF) assigns today's prices to risky future payoffs at alternative investment horizons. It accomplishes this by simultaneously discounting the future and adjusting for risk.
A familiar representation of the SDF is a factorization into the factor discounting at the short-term risk-free interest rate and a martingale component adjusting for risk. This martingale accomplishes the change of probabilities from the data-generating (physical) measure ${\mathbb P}$ to the risk-neutral measure ${\mathbb Q}$.
More recently, \citet{alvarez_2005using}, \citet{hansen_2008consumption}, \citet{hansen_2009}, \citet{hansen_2012} and  \citet{linetsky_2014long} study an alternative {\em long-term factorization} of the SDF. The long-term factorization decomposes the pricing kernel (PK) process in an arbitrage-free asset pricing model
$$
S_t=e^{-\lambda t}\frac{1}{\pi_t}M_t^\infty
$$
into discounting at the long-term discount rate
$\lambda$ (yield on the {\em long bond}, a zero-coupon bond of asymptotically long maturity), a process $\pi_t$ characterizing gross holding period returns on the long bond net of the long-term discount rate, and a positive martingale $M_t^\infty$ that defines a {\em long-term forward measure} ${\mathbb L}$. The process $B_t^\infty=e^{\lambda t}\pi_t$ tracks the gross return earned on the long bond from time zero to time $t$. Then the SDF from time $t+\tau$ to time $t$ takes the form
$$
\frac{S_{t+\tau}}{S_t}=\frac{1}{R^\infty_{t,t+\tau}}\frac{M^\infty_{t+\tau}}{M^\infty_t},
$$
where
$$
\frac{1}{R^\infty_{t,t+\tau}}=\frac{B_{t}^\infty}{B_{t+\tau}^\infty}=e^{-\lambda \tau}\frac{\pi_{t}}{\pi_{t+\tau}}
$$
is the discount factor discounting at the rate of return earned on holding the long bond between times $t$ and $t+\tau$, and the factor $M^\infty_{t+\tau}/M^\infty_t$ encodes the risk adjustment. \citet{alvarez_2005using} originally introduced the long-term factorization in discrete-time ergodic economies. \citet{hansen_2009} introduced and studied the long-term factorization in continuous-time Markovian economies and expressed it in terms of the Perron-Frobenius principal eigenfunction of the pricing operator. Recently  \citet{linetsky_2014long} extended the long-term factorization to general semimartingale economies.
Their martingale approach to the characterization of long-term pricing does not require a Markov specification and is based on a limiting procedure, constructing the long forward measure ${\mathbb L}$ defined by the martingale $M_t^\infty$ as the limit of $T$-maturity forward measures ${\mathbb Q}^T$ familiar in mathematical finance (\citet{jarrow_1987pricing}, \citet{jamshidian_1989exact}, \citet{geman_1995changes}) as maturity increases.
The long-term discount rate $\lambda$ and the process  $\pi_t$ are counterparts of the Perron-Frobenius eigenvalue and eigenfunction of  \citet{hansen_2009}  in the sense that in Markovian economies the process $\pi_t$ reduces to the function of the Markovian state, $\pi(X_t)$, where $\pi(x)$ is the Perron-Frobenius eigenfunction of the pricing operator with the eigenvalue $e^{-\lambda t}$ as in \citet{hansen_2009} (see also \citet{linetsky_2014_cont} for further details on Markovian models).

The long-term factorization of the SDF is convenient in applications to the pricing of long-lived assets  and to theoretical and empirical investigations of the term structure of the risk-return trade-off. In addition to the references above, the growing literature on the long-term factorization and its applications includes \citet{hansen_2012pricing},  \citet{hansen_2013}, \citet{borovicka_2014mis}, \citet{borovivcka2011risk}, \citet{borovivcka2016term},
\citet{bakshi_2012}, \citet{bakshia2015recovery}, \citet{christensen2014nonparametric},  \citet{christensen_2013estimating}, \citet{linetsky_2014_cont}, \citet{linetsky2016bond}, \citet{backus2015term}, \citet{filipovic2016linear}, \citet{filipovic2016relation}, \citet{lustig2016nominal}. Empirical investigations in this literature show that the martingale $M_t^\infty$ is highly volatile and economically significant.  \citet{bakshi_2012} provide theoretical and empirical bounds on the volatility of the martingale component.  \citet{christensen2014nonparametric} estimates the long-term factorization in a structural asset pricing model connecting to the macro-economic fundamentals.  \citet{linetsky2016bond} estimate the long-term factorization in a dynamic term structure model (DTSM) and show how the martingale component in the long-term factorization controls the term structure of the risk-return trade-off in the bond market. In particular, they directly estimate volatility of the martingale $M^\infty$ in a parametric DTSM and show how this martingale gives rise to the downward slopping term structure of bond Sharpe ratios. \citet{lustig2016nominal} apply long-term factorization to the study of foreign exchange markets.

While \citet{linetsky_2014long} provide the theoretical framework and an abstract sufficient condition for existence of the long-term factorization in general semimartingale models without the Markovian assumption, so far only Markovian model specifications have been investigated in the literature.
The purpose of this paper is to construct the long-term factorization of Heath-Jarrow-Morton term structure models, thereby illustrating how the long-term factorization plays out in non-Markovian models.
We adopt the point of view of \citet{filipovic_2001consistency}, \citet{carmona2007interest} and Bjork and view forward curves as elements of an appropriately specified function space. In particular, we follow the specification of \citet{filipovic_2001consistency}. In Section \ref{example_HJM}, after a review of the HJM modeling framework in the setting of \citet{filipovic_2001consistency}, we give a sufficient condition on the asymptotic behavior of the forward rate volatility that ensures existence of the long bond process, the long forward measure, and the long-term factorization in HJM models. This sufficient condition is quite natural from the interest rate modeling point of view and yields existence of the volatility process for the long bond.
Our theoretical results are summarized in Theorem 2, which constitutes the main result of this paper. The proof is given in the Appendix.
In Section \ref{gauss_HJM} we illustrate our results on examples of  (generally non-Markovian) Gaussian HJM models, where our assumptions and their implications can be seen in a transparent way.

The explicit construction of the long-term factorization in this paper furnishes an alternative mechanism of how the long-term factorizations arises, relative to the original theory of \citet{hansen_2009}. Their original formulation of the long-term factorization is based on Markov process theory, and sufficient conditions rely on ergodicity assumptions that furnish the principal eigenvalue and eigenfunction of the pricing semigroup germane to the long-term behavior. \citet{linetsky_2014long} give a general sufficient condition for the existence of the long-term limit in general semimartingale models and show how results of Hansen-Scheinkman arise when the information filtration is Markovian and the pricing kernel is a multiplicative functional of the Markov process generating the filtration.  In contrast to these references, the explicit construction in the present paper illustrates the mechanism of how the long-term factorization arises in non-Markovian HJM models by imposing a condition on the asymptotic behavior of the volatility of the forward curve. The condition is fully explicit, and the proof shows that under this condition HJM models verify the abstract sufficient condition of  \citet{linetsky_2014long} in general semimartingale models. Furthermore, our Gaussian example shows how in the special case of a Gaussian model with constant parameters the Hansen-Scheinkman principal eigenfunction construction of the long-term factorization is recovered from the assumption about asymptotic behavior of forward curve volatility. For empirical analysis of the long-term factorization on US treasury data we refer to \citet{linetsky2016bond}.

\section{Long Term Factorization in Heath-Jarrow-Morton Models}
\label{example_HJM}

The classical \citet{hjm_1992} framework assumes that
the family of zero-coupon bond processes  $\{(P_t^T)_{t\in [0,T]},T\geq 0\}$ is sufficiently smooth across the maturity parameter $T$ so that there exists a family of instantaneous forward rate processes $\{(f(t,T))_{t\in [0,T]},T\geq 0\}$ such that
$P_t^T=e^{-\int_t^T f(t,s)ds},$
and for each maturity $T$ the forward rate is assumed to follow an It\^{o} process
on the time interval $[0,T]$ driven by an $n$-dimensional Brownian motion.
An alternative point of view on interpreting HJM models is to treat $f_t$ as a stochastic process taking values in an appropriate function space of well-defined forward curves. To this end, the \citet{musiela_1993stochastic} parameterization $f_t(x):=f(t,t+x)$ of the forward curve is convenient. Here $x$ denotes time remaining to maturity, so that $t+x$ is the maturity date. This point of view also allows the volatility function to depend on the entire forward curve. In this approach we work with the process $(f_t)_{t\geq 0}$ taking values in an appropriate space of functions on ${\mathbb R}_+$ (the state $f_t$ is a function of time to maturity $x$, $f_t(x)$, with $x\in {\mathbb R}_+$). Mathematical foundations of stochastic processes taking values in function spaces can be found in \citet{da_2014stochastic}, and the development of this point of view in interest rate modeling can be found in \citet{bjork1999interest}, \citet{bjork1999minimal} and \citet{bjork2001existence}, and \citet{carmona2007interest} and  \citet{filipovic_2001consistency}. In this paper we follow the treatment of  \citet{filipovic_2001consistency}.

The HJM forward curve dynamics reads:
\be
df_t=(Df_t+\mu_t)dt+\sigma_t \cdot dW^{\mathbb{P}}_t.
\eel{HJM_function}
The infinite-dimensional standard Brownian motion $W^{\mathbb{P}}=\{(W^{\mathbb{P},j}_t)_{t\geq 0},j=1,2,\ldots\}$ is a sequence of independent standard Browian motions adapted to the underlying reference filtration $({\mathscr F}_t)_{t\geq 0}$ on the probability space $(\Omega,{\mathscr F},{\mathbb P})$. Here
$\sigma_t \cdot dW^\mathbb{P}_t=\sum_{j\in \mathbb{N}} \sigma^j_t dW_t^{\mathbb{P},j}$. The finite-dimensional case arises by simply setting $\sigma_t^j\equiv 0$ for all $j>n$ for some $n$.
The forward curve $(f_t)_{t\geq 0}$ is a process taking values in the Hilbert space $H_w$ that we will define shortly.
The drift  $\mu_t=\mu(t,\omega,f_t)$ and volatility $\sigma_t^j=\sigma^j(t,\omega,f_t)$
take values in the same Hilbert space $H_w$ and depend on
$\omega$ and $f_t$ (note the difference in $w$ and $\omega$; the former is the weight function in the definition of the Hilbert space, while the latter is an element of the sample space $\Omega$). To lighten notation we often do not show dependence of coefficients on $\omega$ and $f$ explicitly. The additional term $Df_t$ in the drift in Eq.\eqref{HJM_function} arises from Musiela's parameterization, where the operator $D$ is interpreted as the first derivative with respect to time to maturity, $Df_t(x) = \partial_x f_t(x)$, and is defined more precisely below as an operator in the Hilbert space $H_w$.


Following \citet{filipovic_2001consistency}, we next define the Hilbert space $H_w$ of forward curves and give conditions on the volatility and drift to ensure that the solution of the HJM evolution equation Eq.\eqref{HJM_function} exists in the appropriate sense (so that the forward curve stays in its prescribed function space $H_w$ as it evolves in time) and specifies an arbitrage-free term structure.
Let $w:\mathbb{R}_{+}\rightarrow[1,\infty)$ be a non-decreasing $C^1$ function such that
\be
\int_0^\infty w^{-1/3}(x)dx<\infty.
\eel{cond_w}
We define
\be
H_w:=\{h\in L_{\text{loc}}^1(\mathbb{R}_+)\ |\ \exists h'\in L_{\text{loc}}^1(\mathbb{R}_+) \enskip\text{and}\enskip \|h\|_w<\infty\},
\ee
where $$\|h\|_w^2:=|h(0)|^2+\int_{R_+} |h'(x)|^2 w(x)dx,$$
and $h'(x)$ is the weak derivative. That is, $H_w$ is defined as the space of locally integrable functions on $\mathbb{R}_+$, with locally integrable weak derivatives, and with the finite norm $\|h\|_w$. Elements of $H_w$ are equivalence classes. Recall that if $h\in L_{\text{loc}}^1(\mathbb{R}_+)$ has a weak derivative $h'\in L_{\text{loc}}^1(\mathbb{R}_+)$, then there exists an absolutely continuous representative of the equivalence class $h$ such that $h(x)-h(y)=\int_y^x h'(z)dz$. Thus, elements of $H_w$ have absolutely continuous representatives. We identify all financial quantities of interest, such as the forward curve, with absolutely continuous representatives. With some abuse of notation, in what follows we do not make the distinction between elements of $H_w$ that are equivalence classes and their absolutely continuous representatives.
The finiteness of the $H_w$-norm imposes tail decay on the derivative $h'$ of the function (the forward curve) with respect to time to maturity such that it decays to zero as time to maturity tends to infinity sufficiently fast so the derivative is square integrable with the weight function $w$, which is assumed to grow sufficiently fast so that  $\int_0^\infty w^{-1/3}(x)dx<\infty$.
By H{\"o}lder's inequality, it is easily seen that
$\int_{\mathbb{R}_+}|h'(x)|dx<\infty$ for all $h\in H_w$.
Thus, the absolutely continuous representative $h(x)$ converges to the limit $h(\infty)\in\mathbb{R}$ as $x\rightarrow\infty$, which can be interpreted as the {\em long forward rate}.  In other words, all forward curves in $H_w$ flatten out sufficiently fast at asymptotically long maturities.
We thus have the following result.
\begin{proposition}
\label{constant_long}
If the initial forward curve $f_0\in H_w$, then there exists a constant $\lambda$ such that
\be
\lim_{T\rightarrow \infty}P_0^{T-t}/P_0^T=e^{\lambda t},
\eel{Plimit}
where $\lambda=f_0(\infty)$ is the long forward rate.
\end{proposition}
Recall that we always identify forward curves with absolutely continuous representatives of elements of $H_w$. Denote the limiting value of the absolutely continuous forward curve $f_0(\infty)=\lambda$. Then Eq.\eqref{Plimit} immediately follows from the relationship between bond prices and forward rates, $P_t^T=e^{-\int_t^T f_t(u)du}$.

We note that \citet{linetsky_2014long} {\em assume} that the initial forward curve satisfies Eq.\eqref{Plimit} to derive a long-term factorization for semimartingale pricing kernels with the long bond factor in the form $B_t^\infty=e^{\lambda t}\pi_t$.  Here in the context of HJM models Proposition 1 is simply an immediate consequence of the Hilbert space structure assumed for  forward curves. From the financial point of view, the property \eqref{Plimit} is natural and only requires that the initial forward curve flatten out at asymptotically long maturities. While in empirical data we do not observe forward curves at asymptotically long maturity, we do typically observe that the term structure becomes flatter between 20 year and 30 years.
We stress that this behavior of the initial forward curve does not impose any restrictions on the term structure dynamics,  as opposed to, for instance, the ``low variance martingale" (LVM) assumption common in the literature on swap market models (cf. \citet{gaspar2016swap}).

The space $H_w$ equipped with $\|h\|_w$ is a separable Hilbert space (Theorem 5.1.1 in \citet{filipovic_2001consistency}).
Define a semigroup of translation operators on $H_w$ by $(T_tf)(x)=f(t+x).$
By \citet{filipovic_2001consistency} Theorem 5.1.1, it is strongly continuous in $H_w$, and we denote its infinitesimal generator by $D$. This is the operator that appears in the drift in Eq.\eqref{HJM_function} due to the Musiela re-parameterization.

We next give conditions on the drift and volatility such that the forward curve stays in $H_w$ as it evolves according to the HJM dynamics. First we need to introduce some additional notation. Define the subspace $H_w^0\subset H_w$ by
$H_w^0=\{f\in H_w \enskip\text{such that }f(\infty)=0\}.$
For any continuous function $f$ on $\mathbb{R}_+$, define a continuous function ${\cal S}f:\mathbb{R}_+\rightarrow\mathbb{R}$ by
\[
({\cal S}f)(x):=f(x)\int_0^x f(\eta)d\eta,\enskip x\in\mathbb{R}_+.
\]
This operator is used to conveniently express the celebrated HJM arbitrage-free drift condition.
By \citet{filipovic_2001consistency} Theorem 5.1.1, there exists a constant $K$ such that $\|{\cal S}h\|_{w}\leq K\| h\|_{w}^2$ for all $h\in H^0_w$.
Local Lipschitz property of $\cal S$ is proved in \citet{filipovic_2001consistency} Corollary 5.1.2, which is used to ensure existence and uniqueness of solution to the HJM equation. Namely,
${\cal S}$ maps $H_w^0$ to $H_w^0$ and is locally Lipschitz continuous:
\be
\|{\cal S}g-{\cal S}h\|_w\leq C(\| g\|_w+\| h\|_w)\| g-h\|_w, \forall g,h\in H_w^0,
\ee
where the constant $C$ only depends on $w$.

Next consider $\ell^2$, the Hilbert space of square-summable sequences, $\ell^2=\{v=(v_j)_{j\in\mathbb{N}}\in\mathbb{R}^\mathbb{N}|\| v\|_{\ell^2}^2:=\sum_{j\in\mathbb{N}}|v_j|^2<\infty\}.$
Let $e_j$ denote the standard orthonormal basis in $\ell^2$. For a separable Hilbert space $H$, let $L_2^0(H)$ denote the space of Hilbert-Schmidt operators from $\ell^2$ to $H$
with the Hilbert-Schmidt norm
$\|\phi\|^2_{L_2^0(H)}:=\sum_{j\in\mathbb{N}}\|\phi^j\|_{H}^2<\infty,$
where $\phi^j:=\phi e_j$. We shall identify the operator $\phi$ with its $H$-valued coefficients $(\phi^j)_{j\in\mathbb{N}}$.

We are now ready to give conditions on the HJM market price of risk and HJM volatility to ensure that the forward curve stays in the Hilbert space $H_w$.
Recall that we have a filtered probability space $(\Omega,{\mathscr F},({\mathscr F}_{t})_{t\geq 0},{\mathbb P})$. Let $\cal P$ denote the predictable sigma-field. For any metric space $G$, we denote by ${\cal B}(G)$ the Borel sigma-field of $G$.
\begin{assumption}
\label{HJM_assumption}
({\bf Conditions on Volatility, Market Price of Risk and Initial Forward Curve})\\
(i) The initial forward curve $f_0\in H_w$.\\
(ii) The (negative of the) market price of risk $\gamma$ is a measurable function from $(\mathbb{R}_+\times\Omega\times H_w,{\cal P}\otimes{\cal B}(H_w))$ into $(\ell^2, {\cal B}(\ell^2))$ such that there exists a function $\Gamma\in L^2(\mathbb{R}_+)$ that satisfies
\be
\|\gamma(t,\omega,h)\|_{\ell^2}\leq\Gamma(t)\,\  {\rm for\ all}\,\  (t,\omega,h).
\eel{D3}
(iii) The volatility $\sigma=(\sigma^j)_{j\in\mathbb{N}}$ is a measurable function from $(\mathbb{R}_+\times\Omega\times H_w,{\cal P}\otimes{\cal B}(H_w))$ into $(L_2^0(H_w^0),{\cal B}(L_2^0(H_w^0)))$. It is is assumed to be Lipschitz continuous in $h$ and uniformly bounded, i.e. there exist constants $D_1, D_2$ such that for all $(t,\omega)\in\mathbb{R}_+\times\Omega$ and $h,h_1,h_2\in H_\omega$
\be
\|\sigma(t,\omega,h_1)-\sigma(t,\omega,h_2)\|_{L_2^0(H_w)}\leq D_1\|h_1-h_2\|_{H_w},\quad \|\sigma(t,\omega,h)\|_{L_2^0(H_w)}\leq D_2.
\eel{cond_sigma}
\end{assumption}
In the case when $W^\mathbb{P}$ is finite-dimensional, simply replace $\ell^2$ with ${\mathbb R}^n$.
The drift $\mu_t=\mu(t,\omega,f_t)$ in \eqref{HJM_function} is defined by the HJM drift condition, that takes the form in our notation
$$\mu(t,\omega,f_t)=\alpha^{\text{HJM}}(t,\omega,f_t)-(\gamma \cdot \sigma)(t,\omega,f_t),$$ where
$$\alpha^{\text{HJM}}(t,\omega,f_t)=\sum_{j\in\mathbb{N}}{\cal S}\sigma^j(t,\omega,f_t).$$

The following theorem summarizes the properties of the HJM model  \eqref{HJM_function} in this setting (see \citet{filipovic_2001consistency} Theorem 5.2.1).
\begin{theorem}{\bf (HJM Model)}
\label{exist_HJM}
(i) Eq.\eqref{HJM_function} has a unique continuous weak solution.\\
(ii) For each $t\geq 0$
$$f_t(\infty)=f_0(\infty).$$
(iii) The pricing kernel has the risk-neutral factorization
$$S_t=\frac{1}{A_t}M_t$$
with the implied savings account $A_t$, the martingale $M_t$, and the risk-neutral measure  given by:
\be
A_t=\exp\big(\int_0^t f_s(0)ds\big),
\ee
\be
M_t=\exp\Big(\int_0^t \gamma_s \cdot dW^{\mathbb{P}}_s-\frac{1}{2}\int_0^t \|\gamma_s\|_{\ell^2}^2ds\Big),\,
\mathbb{Q}|_{\mathscr{F}_t}= M_t \mathbb{P}|_{\mathscr{F}_t}.
\ee
The process
$W_t^{\mathbb{Q}}:=W_t^{\mathbb{P}}-\int_0^t \gamma_s ds$
is an (infinite-dimensional) Brownian motion under $\mathbb{Q}$. \\
(iv) The $T$-maturity bond valuation process $P_t^T$ has the form under $\mathbb{P}$:
\[
\frac{P_t^T}{P_0^T}=A_t\exp\Big(\int_0^t \sigma_s^T \cdot \gamma_s  ds- \int_0^t \sigma_s^T\cdot dW^{\mathbb{P}}_s-\frac{1}{2} \int_0^t \|\sigma_s^T\|^2_{\ell^2}ds\Big),\quad t\in [0,T],
\]
where the volatility of the $T$-maturity bond is
$$\sigma^{T}_t=\int_0^{T-t}\sigma_t(u)du, \quad t\in [0,T].$$
The process under ${\mathbb Q}$ reads:
\be
\frac{P_t^T}{P_0^T}=A_t\exp\Big(- \int_0^t \sigma_s^T\cdot dW^{\mathbb{Q}}_s-\frac{1}{2} \int_0^t \|\sigma_s^T\|^2_{\ell^2}ds\Big),\quad t\in [0,T].
\ee
\end{theorem}
For the definition of a weak solution used here see \citet{filipovic_2001consistency} Definition 2.4.1. The proof of Theorem \ref{exist_HJM} follows from the results in \citet{filipovic_2001consistency} and  is summarized for the readers' convenience in  Appendix.

We  are now ready to formulate the long-term factorization in HJM models.
First we observe that by Theorem \ref{exist_HJM} the product of the pricing kernel and the gross return on the $T$-maturity bond
\be
M_t^T:=S_t\frac{P_t^T}{P_0^T}=\exp\left(-\frac{1}{2}\int_0^t\|\gamma_s-\sigma_s^T\|_{\ell^2}^2ds+\int_0^t(\gamma_s-\sigma_s^T)\cdot dW_s^\mathbb{P}\right)
\ee
is a positive $\mathbb{P}$-martingale on $t\in[0,T]$ starting at unity, $M_0^T=1$. We can use it to define a new probability measure $\mathbb{Q}^T$ on $\mathscr{F}_T$ by $\mathbb{Q}^T|_{\mathscr{F}_T}=M_T^T\mathbb{P}|_{\mathscr{F}_T}$. $\mathbb{Q}^T$ is the $T$-forward measure originally introduced by \citet{jarrow_1987pricing} and by \citet{geman_1989importance} and \citet{jamshidian_1989exact}. Under ${\mathbb Q}^T$ the $T$-maturity zero-coupon bond serves as the numeraire.
We are interested in the long-term limit $T\rightarrow\infty$. Taking the limit naively in the expression below and writing $M_t^\infty:=S_t\frac{P_t^\infty}{P_0^\infty}$ will not generally work because $P_t^\infty$ will typically vanish due to discounting over an infinite horizon. Nevertheless, the limit of the ratio, $\lim_{T\rightarrow\infty}\frac{P_t^T}{P_0^T}$, can be made precise in general semimartingale models.  \citet{linetsky_2014long} define this limit in Emery's semimartingale topology (see \citet{emery_1979topologie} and \citet{linetsky_2014long} for the definition of Emery's distance). To this end, it is first convenient to extend the process $P_t^T/P_0^T$ to all $t\in[0,\infty)$ beyond $[0,T]$ by considering a self-financing roll-over strategy that starts at time zero by investing one unit of account in $1/P_{0}^T$ units of the $T$-maturity zero-coupon bond. At time $T$ the bond matures, and the value of the strategy is $1/P_{0}^T$ units of account. We roll the proceeds over by re-investing into $1/(P_{0}^T P_{T}^{2T})$ units of the zero-coupon bond with maturity $2T$. We continue with the roll-over strategy, at each time $kT$ re-investing the proceeds into the bond $P_{kT}^{(k+1)T}$. We denote the wealth process of this self-financing strategy by $B_t^T$,
$B_t^T = \left(\prod_{i=0}^k P_{iT}^{(i+1)T}\right)^{-1} P_{t}^{(k+1)T},$ $t\in [kT,(k+1)T),$ $k=0,1,\ldots.$
It is clear that $B_t^T$ extends $P_t^T/P_0^T$ to all times $t\geq0$.
As a consequence, the product $S_tB_t^T$ also extends the martingale $M_t^T$ to all times $t\geq0$. We continue to use the notation $M_t^T$ for $S_tB_t^T$. $M_t^T$ now defines a new probability measure for all $t\geq0$, and we still denote it by $\mathbb{Q}^T$.
With these preparations completed, we are now ready to formulate the long-term factorization in HJM models.

\begin{theorem}{\bf (Long-Term Factorization in HJM Models)}
\label{main_HJM}
Suppose the initial forward curve $f_0$ and the market price of risk $\gamma_t$ satisfy Assumption \ref{HJM_assumption} (i) and (ii). Suppose the
volatility $\sigma=(\sigma^j)_{j\in\mathbb{N}}$ is a measurable function from $(\mathbb{R}_+\times\Omega\times H_w,{\cal P}\otimes{\cal B}(H_w))$ into $(L_2^0(H_{\bar{w}}^0),{\cal B}(L_2^0(H_{\bar{w}}^0)))$ and is Lipschitz continuous in $h$ and uniformly bounded as in Assumption \ref{HJM_assumption} (ii), where
$H_{\bar{w}}^0\subseteq H_w^0$ with $\bar{w}$ satisfying $\int_0^\infty \bar{w}^{-1/3}(x)dx<\infty$  and having the large-$x$ asymptotics:
\be
\frac{1}{\bar{w}(x)}=O(x^{-(3+\epsilon)})
\eel{wprime}
for some $\epsilon>0$. Then the following results hold. \\
(i) $(M_t^T)_{t\geq0}$ converge to a positive martingale $M_t^\infty$ in Emery's semimartingale topology. $(B_t^T)_{t\geq0}$ converge to a positive process $B_t^\infty$ in Emery's semimartingale topology. $\mathbb{Q}^T$ converge to a limiting measure $\mathbb{Q}^\infty$ in total variation norm.\\
(ii) The HJM pricing kernel admits the long-term factorization
\be
S_t=e^{-\lambda t}\frac{1}{\pi_t}M_t^\infty,
\ee
where $B_t^\infty=e^{\lambda t}\pi_t$ is the long bond process.
\\
(iii) The process
\be
\sigma_t^\infty:=\int_0^\infty \sigma_t(u)du
\eel{sigmainf}
is well defined, and the long bond process
$B_t^\infty$ satisfies:
\be
B_t^\infty=A_t\exp\big(\int_0^t \sigma_s^\infty \cdot \gamma_s  ds-\int_0^t \sigma_s^\infty\cdot dW^{\mathbb{P}}_s-\frac{1}{2} \int_0^t \|\sigma_s^\infty \|_{\ell^2}^2ds\big)
\eel{HJM_longbond}
with volatility $\sigma_t^\infty$.\\
(iv) The martingale $M_t^\infty$ satisfies:
\be
M_t^\infty=\exp\Big(\int_0^t \gamma_s^\infty dW_s^\mathbb{P} -\frac{1}{2}\int_0^t \|\gamma^\infty_s\|_{\ell^2}^2ds \Big),
\eel{HJMMinfty}
where the market price of risk is
$$\gamma_t^\infty=\gamma_t-\sigma_t^\infty.$$
(v) The measure $\mathbb{Q}^\infty$ is given by
$\frac{d\mathbb{Q}^\infty}{d\mathbb{P}}|_{\mathscr{F}_t}=M_t^\infty.$
Under $\mathbb{Q}^\infty$,
$W^{\mathbb{Q}^\infty}_t:=W^{\mathbb{P}}_t-\int_0^t \gamma_s^\infty ds$
is a standard Brownian motion, and the $\mathbb{Q}^\infty$-dynamics of the forward curve, the pure discount bond, and the long bond are:
\be
df_t=(Df_t+\alpha^{HJM}_t-\sigma_t^{\infty}\cdot \sigma_t)dt+\sigma_t\cdot dW^{\mathbb{Q}^\infty}_t,
\ee
\be
\frac{P_t^T}{P_0^T}=A_t\exp\Big(\int_0^t \sigma_s^T \cdot \sigma^\infty_s  ds- \int_0^t \sigma_s^T\cdot dW^{\mathbb{Q}^\infty}_s-\frac{1}{2} \int_0^t \|\sigma_s^T\|_{\ell^2}^2ds\Big),\quad t\in [0,T],
\ee
\be
B_t^\infty=A_t\exp\big(-\int_0^t \sigma_s^\infty\cdot dW^{\mathbb{Q}^\infty}_s+\frac{1}{2} \int_0^t \|\sigma_s^\infty \|_{\ell^2}^2ds\big),\quad t\geq 0.
\ee
\end{theorem}

The proof is given in Appendix. The long bond $B_t^\infty$ is the gross return from holding a zero-coupon bond of asymptotically infinite maturity from time $0$ to time $t$. $\mathbb{Q}^\infty$ is termed the long forward measure (also denoted as $\mathbb{L}$ in \citet{linetsky_2014long}) as the limit of the $T$-forward measure. The sufficient condition on the forward curve volatility to ensure existence of the long bond, long-term factorization and long forward measure is a strengthening of Filipovic's condition on forward curve volatility in Assumption \ref{HJM_assumption} for existence of a solution of the HJM SDE \eqref{HJM_function} in $H_w$. As a sufficient condition for existence of the long-term factorization, we require that the weight function $\bar{w}$ in the weighted Sobolev space $H_{\bar{w}}$ where the volatility components take their values satisfy the asymptotics \eqref{wprime}, a strengthening of Filipovic's assumption on the weight $w$ in the definition of the norm of the space $H_w$ where the forward curves themselves evolve. Alternatively, we can assume that $w$ satisfies the asymptotics \eqref{wprime} and work with this smaller function space from the beginning. However, this is not necessary, and so we leave the assumptions on the space of forward curves unchanged and the same as in Filipovic, while imposing the sufficient condition on forward rate volatility to ensure existence of the long-term factorization.
Typical choices of weight function are $w(x)=e^{\alpha x}$ for $\alpha>0$ and $w(x)=(1+x)^\alpha$ for $\alpha>3$, which both satisfy the asymptotics \eqref{wprime}. An example that satisfies Eq.\eqref{cond_w} but not \eqref{wprime} is $w(x)=(1+x)^3 \left(\log (2+x)\right)^6$.

We note that the mechanism that ensures existence of the long-term limit in HJM models is the combination of the sufficiently fast flattening of the forward curve at long maturities that is ensured by the structure of the space $H_w$, as well as sufficiently fast decay of the forward curve volatility for long maturities that is ensured by the structure of the space $H_{\bar{w}}$. In particular, the forward curve flattens sufficiently fast for long maturities so that the long forward rate process $f_t(\infty)$ exists and is constant (part (ii) in Theorem \ref{exist_HJM}). The forward curve volatility decays fast enough that the long bond volatility given by the integral \eqref{sigmainf} is well defined. Verification that the volatility of the long bond \eqref{sigmainf} is well defined constitutes the key part of the proof (see Appendix).

Theorem \ref{main_HJM} provides a fully explicit construction of the long-term factorization in HJM models driven by an infinite-dimensional Brownian motion. The existence of the long bond is solely determined by the dynamics of forward curve. To ensure existence of the long bond, a sufficient condition is imposed on the volatility of the forward curve. An obvious necessary condition is that the long forward rate needs to be constant (if it exists), as otherwise the long bond will instantaneously decrease to zero when there is a (necessary positive due to the theorem of \citet{dybvig_1996long}) shock to the long forward rate. The framework of \citet{filipovic_2001consistency} already ensures this necessary condition by restricting the space of forward curves and forward curve volatilities. Theorem 2 gives a sufficient condition on the space of volatility curves that ensures existence of the long bond and the long-term factorization of the HJM pricing kernel.
In particular, Theorem \ref{main_HJM} yields an explicit decomposition of the market price of risk in HJM models
$\gamma_t = \sigma_t^\infty+ \gamma_t^\infty$
into a component identified with the volatility of the long bond $\sigma^\infty_t$ and a component $\gamma_t^\infty$ defining the martingale $M^\infty_t$and, in turn, the long forward measure ${\mathbb{Q}^\infty}$.

\section{Example: Gaussian HJM Models}
\label{gauss_HJM}

Assume the initial forward curve $f_0\in H_w$ with some weight $w$ satisfying \eqref{cond_w}.
When the forward curve volatility is deterministic (independent of $\omega$ and the forward curve $f_t$), the conditions on volatility in Theorem \ref{main_HJM} simplify to requiring that $\sigma(t,\omega,h)=\sigma(t)\in L_2^0(H^0_{\bar{w}})$   for some weight $\bar{w}$ such that  $1/\bar{w}(x)=O(x^{-(3+\epsilon)})$ and that $\sigma(t)$ is uniformly bounded.
Under these assumptions, the forward curve follows a Gaussian process taking values in $H_w$ under both ${\mathbb Q}$ and ${\mathbb{Q}^\infty}$. It also follows a Gaussian process in $H_w$ under ${\mathbb P}$ if the market price of risk $\gamma_t$ (assumed to satisfy \eqref{D3}) is also deterministic. We note that this Gaussian process is generally not Markovian.

We now consider a special case with
\be
\sigma^j(t)(x) = \sigma^j(x)=\sigma_j e^{-\kappa_j x},
\eel{vasicekvol}
where $\sigma_j\geq0$, $\kappa_j\geq 0$ and $\sum_{j=1}^\infty \sigma_j^2(1+\kappa_j)<\infty$.

Let $\bar{w}(x)=x^{-4}\wedge1$. Then
\be
\begin{array}{ll}
\|\sigma_t\|^2_{L_2^0(H_{\bar{w}})} & = \displaystyle{\sum_{j=1}^\infty \|\sigma_t^j\|^2_{H_{\bar{w}}}}\\
& = \displaystyle{\sum_{j=1}^\infty \sigma_j^2\left(1+\int_0^\infty \kappa_j^2 e^{-2\kappa_jx}\bar{w}(x)dx\right)}\\
& \leq \displaystyle{\sum_{j=1}^\infty \sigma_j^2  \left(1+\int_0^\infty \kappa_j^2 e^{-2\kappa_jx}dx\right)}\\
& = \displaystyle{\sum_{j=1}^\infty \sigma_j^2 (1+\kappa_j/2)}<\infty.\\
\end{array}
\ee
Thus, $\sigma_t$ satisfies \eqref{cond_sigma}. This ensures the model satisfies all assumptions in Theorem \ref{main_HJM} and all results in Theorem \ref{main_HJM} hold.

To simplify notation, consider the scalar case with $\sigma_1>0$ and $\sigma_j=0$ for all $j>1$ and drop the index $1$ in $\sigma_1$, which is the so-called extended Vasicek model also known as the Hull-White model. The solution to the SDE \eqref{HJM_function} under the risk-neutral measure (setting the market price of risk $\gamma_t$ to zero) is explicit (cf. Carmona and Tehranci p.178):
$$
f_t(x)=f_0(t+x)+\frac{\sigma^2}{\kappa^2}e^{-\kappa x}(1-e^{-\kappa t})(1-e^{-\kappa x}(1+e^{\kappa t})/2)+\sigma e^{-\kappa x}\int_0^t e^{-\kappa(t-s)}dW^{\mathbb Q}_s.
$$
In the limit $x\rightarrow \infty$, we explicitly obtain the constant long forward rate:
$$
f_t(\infty)=f_0(\infty)=:\lambda
$$
for all $t\geq 0$ (the initial forward curve $f_0\in H_w$ possesses a long forward rate $f_0(\infty)$, and it is preserved in time under the HJM evolution as the forward curve evolves). We stress that this example is time-inhomogeneous in general, and so here $\lambda$ is defined as the limiting value of the forward rate and is {\em not} the principal eigenfunction of the pricing semigroup as in the Markovian case in \citet{hansen_2009}.

On the other hand, letting $x=0$, we have the extended Vasicek evolution for the short rate $r_t=f_t(0)$:
$$
r_t=f_0(t)+\frac{\sigma^2}{2\kappa^2}(1-e^{-\kappa t})^2+\sigma \int_0^t e^{-\kappa (t-s)}dW^{\mathbb Q}_s,
$$
which satisfies the SDE
$dr_t =\kappa(\theta_{\mathbb Q}(t)-r_t)dt + \sigma dW_t^{\mathbb Q}$
with the time-dependent parameter
$$
\theta_{\mathbb Q}(t)=\frac{1}{\kappa}f_0^\prime (t)+f_0(t)+\frac{\sigma^2}{2\kappa^2}(1-e^{-2\kappa t}).
$$
If, in particular, we require that $\theta_{\mathbb Q}$ is constant, we then obtain a constraint on the initial forward curve $\frac{1}{\kappa}f_0^\prime (t)+f_0(t)+\frac{\sigma^2}{\kappa^2}(1-e^{-\kappa t})=\theta_{\mathbb Q}$, whose solution is the initial forward curve in the time-homogeneous \citet{vasicek_1977equilibrium} model:
$$
f_0(t)=\left(\theta_\mathbb{Q}\kappa-\frac{\sigma^2(1-e^{-\kappa t})}{2\kappa}\right)\frac{1-e^{-\kappa t}}{\kappa}+r_0e^{-\kappa t}.
$$
Now, in this special case with constant parameters the model is time-homogeneous Markov, and the limiting value of the forward curve indeeds becomes the familiar principal eigenfunction in the Vasicek model:
$$
f_0(\infty)=\lambda=\theta_{\mathbb Q}-\sigma^2/(2\kappa).
$$

Returning to the extended Vasicek model with a general initial forward curve $f_0\in H_w$, the volatility of the long bond is constant:
$$\sigma_t^\infty=\sigma^\infty=\int_0^\infty\sigma e^{-\kappa x}dx=\frac{\sigma}{\kappa}$$
and the long bond has a simple ${\mathbb Q}$-dynamics:
$$
B_t^\infty = A_t e^{-\frac{\sigma}{\kappa} W_t^{\mathbb Q}-\frac{1}{2}\frac{\sigma^2}{\kappa^2}t},
$$
where $A_t=e^{\int_0^t r_s ds}$ is the savings account.

To further illustrate calculations in the simplest possible setting, we now also assume that the market price of risk is constant, $\gamma_t=\gamma$.  In this case $$\gamma^\infty = \gamma-\sigma^\infty=\gamma-\frac{\sigma}{\kappa}$$ is also constant, and $M_t^\infty$ is the exponential ${\mathbb P}$-martingale:
$$
M_t^\infty=e^{\gamma^\infty W^{\mathbb P}_t- \frac{1}{2}(\gamma^\infty)^2 t}.
$$
Then the forward curve has the ${\mathbb{Q}^\infty}$-measure dynamics:
$$
f_t(x)=f_0(t+x)-\frac{\sigma^2}{2\kappa^2}e^{-2\kappa x}(1-e^{-2\kappa t})+\sigma e^{-\kappa x}\int_0^t e^{-\kappa(t-s)}dW^{\mathbb{Q}^\infty}_s,
$$
and in particular for the short rate we obtain:
$$
r_t=f_0(t)-\frac{\sigma^2}{2\kappa^2}(1-e^{-2\kappa t})+\sigma \int_0^t e^{-\kappa (t-s)}dW^{\mathbb{Q}^\infty}_s,
$$
which satisfies the SDE
$dr_t =\kappa(\theta_{\mathbb{Q}^\infty}(t)-r_t)dt + \sigma dW_t^{\mathbb{Q}^\infty}$
with the time-dependent parameter under ${\mathbb{Q}^\infty}$:
$$
\theta_{\mathbb{Q}^\infty}(t)=\frac{1}{\kappa}f_0^\prime (t)+f_0(t)+\frac{\sigma^2}{2\kappa^2}(1+e^{-2\kappa t})=\theta_{\mathbb{Q}}(t)-\frac{\sigma^2}{\kappa^2}.
$$
The long bond ${\mathbb{Q}^\infty}$-dynamics is:
$$
B_t^\infty = A_t e^{-\frac{\sigma}{\kappa} W_t^{\mathbb{Q}^\infty}+\frac{1}{2}\frac{\sigma^2}{\kappa^2}t}.
$$
This example illustrates the long-term factorization in a simple setting of Gaussian models and makes explicit how the market price of Brownian risk $\gamma_t$ is explicitly decomposed into the volatility of the long bond $\sigma_t^\infty$ plus the market price of risk under the long forward measure $\gamma^\infty$ that defines the martingale component $M_t^\infty$ in the long term factorization.
According to the recent empirical evidence in the bond market, the latter component is large and highly economically significant, as it controls the shape of the term structure of bond Sharpe ratios. We refer the reader to \citet{linetsky2016bond}, where a particular Markovian specification is empirically estimated on the US Treasury data. In contrast, this paper provides a general decomposition of the market price of Brownian risk in non-Markovian HJM models and, in particular, in time-inhomogeneous Gaussian models, such as the multi-factor Hull-White-type models popular in practice.

\appendix

\section{Proof of Theorem 2}
To prepare for the proof of Theorem 2, we first briefly sketch the proof of Theorem 1, referring to \citet{filipovic_2001consistency} for details.
\noindent\emph{Proof of Theorem \ref{exist_HJM}.}
(i) We first consider the risk-neutral case with $\gamma=0$:
\be
df_t=(Df_t+\alpha_{\text{HJM}}(t,\omega,f_t))dt+\sum_{j\in\mathbb{N}}\sigma^j(t,\omega,f_t)dW^{\mathbb{Q},j}_t.
\eel{HJM_Q}
By Assumption \ref{HJM_assumption}, $\alpha_{\text{HJM}}(t,\omega,h)$ is Lipschitz continuous in $h$ and uniformly bounded (cf. \citet{filipovic_2001consistency} Lemma 5.2.2). Thus by Theorem 2.4.1 of \citet{filipovic_2001consistency}, Eq.\eqref{HJM_Q} has a unique continuous weak solution.

Uniqueness of a weak solution with non-zero $\gamma$ follows by the application of Girsanov's theorem.
We already have uniqueness of a weak solution with $\gamma=0$.
By \eqref{D3}, $\gamma$ satisfies Novikov's condition (cf. \citet{filipovic_2001consistency} Lemma 2.3.2). Thus, we can define a new measure $\mathbb{P}$ by
\be
\mathbb{P}|_{\mathscr{F}_t}=\exp\Big(-\frac{1}{2}\int_0^t\|\gamma_s\|^2_{\ell^2}ds-\int_0^t\gamma_s\cdot dW^{\mathbb{Q}}_s\Big)\mathbb{Q}|_{\mathscr{F}_t}.
\eel{QP}
Then by Girsanov's theorem for infinite-dimensional Brownian motion (cf. \citet{filipovic_2001consistency})
\be
W^{\mathbb{P}}_t=W^{\mathbb{Q}}_t+\int_0^t\gamma_s ds
\eel{W_QP}
is an infinite-dimensional standard Brownian motions under $\mathbb{P}$. Thus, $f_t$ is a unique weak solution  of the HJM equation \eqref{HJM_function} under $\mathbb{P}$ with general $\gamma$.\\
(ii) Since $\alpha^{\text{HJM}}\in H_w^0$, $f_t(\infty)$ is constant.\\
(iii) By \citet{filipovic_2001consistency} Theorem 5.2.1, zero-coupon bond price processes $(P_t^T/A_t)_{t\geq0}$ taken relative to the process $A_t=e^{\int_0^t f_s(0)ds}$  are $\mathbb{Q}$-martingales. This immediately yields the risk-neutral factorization of the pricing kernel under ${\mathbb P}$.\\
(iv) \citet{filipovic_2001consistency} Eq.(4.17) gives
\be
\frac{P_t^T}{P_0^T}=A_t\exp\Big(- \int_0^t \sigma_s^T\cdot dW^{\mathbb{Q}}_s-\frac{1}{2} \int_0^t \|\sigma_s^T\|_{\ell^2}^2ds\Big).
\ee
Using Eq.\eqref{W_QP} gives the bond dynamics under ${\mathbb P}$. $\Box$

\enskip\\
\emph{Proof of Theorem \ref{main_HJM}.}
We are now ready for the proof of Theorem 2. The proof consists of two parts. We first prove that the processes on the right hand side of \eqref{HJM_longbond} and \eqref{HJMMinfty} are well defined (the integrals in the exponential are well defined).
Next we prove that
\be
\frac{{\mathbb E}^{\mathbb P}_t[S_T]}{{\mathbb E}^{\mathbb P}[S_T]} \xrightarrow{\rm L^1} M_t^\infty\quad \text{as} \quad T\rightarrow \infty
\eel{PKL1}
with $M_t^\infty$ defined by the right hand side of \eqref{HJMMinfty}. By Theorem 3.1 and 3.2 of \citet{linetsky_2014long} and Proposition \ref{constant_long}, (i)-(iv) follows. The expression for $W_t^{\mathbb{Q}^\infty}$ then follows from Girsanov's Theorem (cf. \citet{filipovic_2001consistency} Theorem 2.3.3), and the SDE for $f_t$ under $\mathbb{Q}^\infty$ then follows immediately.

Since $M_t^\infty=S_tB_t^\infty$, we just need to prove right hand side of \eqref{HJM_longbond} is well defined. We first prove the following lemma which is central to all of the subsequent estimates.
\begin{lemma}
\label{HJM_additional1}
The following estimate holds for
any function $h\in H_{\bar{w}}^0$:
$$
\int_T^\infty |h(x)|dx\leq C(T)\|h\|_{\bar{w}},\quad \text{where}  \quad C(T)=K(T^{-\epsilon/2}\wedge 1)
$$
for some $K>0$ and $\epsilon>0$.
\end{lemma}
\begin{proof}
Since $\bar{w}(x)\geq 1$, for all $h\in H_{\bar{w}}^0$ we can write
$$
|h(x)|=\left|\int_x^\infty h'(s)ds\right|\leq \| h\|_{\bar{w}} \left(\int_x^\infty \frac{ds}{\bar{w}(s)}\right)^{1/2}
$$
$$
\leq \|h\|_{\bar{w}}\left(\int_x^\infty K(s^{-(3+\epsilon)}\wedge1)ds\right)^{1/2}
$$
$$
\leq \|h\|_{\bar{w}} K(x^{-(1+\epsilon/2)}\wedge1),
$$
where the constant $K$ can change from step to step. Thus, $\int_T^\infty |h(x)|dx\leq \|h\|_{\bar{w}}K(T^{-\epsilon/2}\wedge1)$.
\end{proof}

Lemma \ref{HJM_additional1} ensures that each element of the vector $\sigma_t^\infty$ in (3.9) is well defined. The next lemma ensures that $\sigma_t^\infty\in{\ell^2}$ and the RHS of \eqref{HJM_longbond} is well defined.
\begin{lemma}
\label{est_sigma}
$\int_0^t\|\sigma_s^\infty\|_{\ell^2}^2ds\leq C^2(0) tD_2^2$.
\end{lemma}
\begin{proof}
By Lemma \ref{HJM_additional1}, $\int_0^\infty |\sigma_s^j(u)|du\leq C(0)\|\sigma_s^j\|_{\bar{w}}$. This implies
\be
\begin{array}{ll}
\displaystyle{\int_0^t\|\sigma_s^\infty\|_{\ell^2}^2ds}
&\displaystyle{\leq \int_0^t\sum_{j\in\mathbb{N}}\left(\int_0^\infty |\sigma_s^j(u)| du\right)^2 ds}\\
&\displaystyle{\leq\int_0^t\sum_{j\in\mathbb{N}}C^2(0) \|\sigma_s^j\|_{\bar{w}}^2 ds}\\
&\displaystyle{=C^2(0)\int_0^t \|\sigma_s\|^2_{L_2^0(H_{\bar{w}})}ds}\\
&\leq C^2(0) tD_2^2,\\
\end{array}
\ee
where $D_2$ is the volatility bound in Eq.\eqref{cond_sigma}.
\end{proof}
By above lemma, the last integral in \eqref{HJM_longbond} is well defined. The stochastic integral $ \int_0^t \sigma_s^\infty \cdot dW^{\mathbb{P}}_s$ is well defined due to It\^{o}'s isometry. The first integral is bounded by
\be
\frac{1}{2}\int_0^t \big(\|\gamma_s\|_{\ell^2}^2+\|\sigma_s^\infty\|_{\ell^2}^2\big)ds\leq\frac{1}{2}\int_0^t\Gamma(s)^2ds+\frac{1}{2}C^2(0)tD_2^2,
\ee which is well defined by the fact that $\Gamma\in L_2(\mathbb{R}_+)$.  Thus the right hand side of \eqref{HJM_longbond} is well defined.

We now turn to the verification of Eq.\eqref{PKL1}.
We first re-write $P_t^T/P_0^T$ and $B_t^\infty$ defined by Eq.\eqref{HJM_longbond} in terms of $\mathbb{Q}$-Brownian motion $W_t^{\mathbb{Q}}$:
\be
\frac{P_t^T}{P_0^T}=A_t\exp\Big(- \int_0^t \sigma_s^T\cdot dW^{\mathbb{Q}}_s-\frac{1}{2} \int_0^t \|\sigma_s^T\|^2_{\ell^2}ds\Big),
\ee
\be
B_t^\infty=A_t\exp\Big(- \int_0^t \sigma_s^\infty\cdot dW^{\mathbb{Q}}_s-\frac{1}{2} \int_0^t\|\sigma_s^\infty\|_{\ell^2}^2ds\Big).
\ee

Fix the current $t\geq 0$. We note that the condition \eqref{PKL1} can be written under any locally equivalent  probability measure ${\mathbb Q}^V$ associated with any valuation process $V$:
\be
\lim_{T\rightarrow\infty}\mathbb{E}^{\mathbb{Q}^V}[|B_t^T/V_t-B_t^\infty/V_t|]=0.
\eel{L1_alter}
We can use this freedom to choose the measure convenient for the setting at hand. Here we choose to verify it under $\mathbb{Q}$, i.e.
\be
\lim_{T\rightarrow\infty}\mathbb{E}^{\mathbb{Q}}\left[\left|\frac{P_t^T}{P_0^TA_t}- \frac{B_t^\infty}{A_t}\right|\right]=0.
\eel{L1_HJM_target}

We first introduce some notation.
For $v\in[0,t]$ and $T\in[t,\infty]$ define the quantities
\begin{gather}
j_v^T:= \int_0^v \sigma_s^T\cdot dW^{\mathbb{Q}}_s,\enskip k_v^T:=\frac{1}{2} \int_0^v \|\sigma_s^T\|_{\ell^2}^2ds,\\
\bar{\sigma}_v^T:=\sigma_v^\infty-\sigma_v^T=\int_{T-v}^\infty \sigma_v(u)du,\enskip z_v^T:=\frac{1}{2}\int_0^v\|\bar{\sigma}_s^T\|_{\ell^2}^2ds,\enskip Y_v^T:=e^{-(j_v^T-j_v^\infty)-z_v^T}.
\end{gather}
For  $p\geq1$ and a random variable $X$ we denote
$$\|X\|_p:=\big(\mathbb{E}^{\mathbb{Q}}[|X|^p]\big)^{1/p},$$ as long as the expectation is well defined.
Then Eq.\eqref{L1_HJM_target} can be re-written as
\be
\lim_{T\rightarrow\infty} \|e^{-j_t^T-k_t^T}-e^{-j_t^\infty-k_t^\infty}\|_1=0.
\eel{L1_HJM_target2}
By H{\"{o}}lder's inequality,
\be
\varlimsup_{T\rightarrow\infty} \|e^{-j_t^T-k_t^T}-e^{-j_t^\infty-k_t^\infty}\|_1  \leq \varlimsup_{T\rightarrow\infty}\|e^{-j_t^\infty-k_t^\infty}\|_2\| e^{-(j_t^T-j_t^\infty)-(k_t^T-k_t^\infty)}-1\|_2.
\ee
Lemma \ref{L1_HJM_lemma1} and \ref{L1_HJM_lemma3} below show that
$\|e^{-j_t^\infty-k_t^\infty}\|_2$ is finite and $$\lim_{T\rightarrow\infty} \|e^{-(j_t^T-j_t^\infty)-(k_t^T-k_t^\infty)}-1\|_2=0,$$ respectively.


\begin{lemma}
\label{L1_HJM_lemma1}
For each $t>0$, there exists $C$ such that
\be
\sup_{v\leq t}\|Y_v^T\|_2\leq C <\infty \text{ and }\|e^{-j_t^\infty-k_t^\infty}\|_2\leq C<\infty.
\ee
\end{lemma}
\begin{proof}
We begin by considering the process
$(Y_v^T)^2=e^{-(2j_v^T-2j_v^\infty)-4z_v^T+2z_v^T}$ for $t\in [0,T]$. By It\^{o}'s formula, $e^{-(2j_v^T-2j_v^\infty)-4z_v^T}$ is a local martingale. Since it is also positive, it is a supermartingale (in fact, it is a true martingale due to Lemma \ref{est_sigma} and Novikov's criterion).
Therefore for all $v\leq t$,
\be
\mathbb{E}^{\mathbb{Q}}[e^{-(2j_v^T-2j_v^\infty)-4z_v^T}]\leq1.
\ee
Similar to Lemma \ref{est_sigma}, $|z_v^T|\leq\frac{1}{2}C^2(T-v)vD_2^2$. Thus $\|Y_v^T\|_2^2=\mathbb{E}^{\mathbb{Q}}[e^{-(2j_v^T-2j_v^\infty)-4z_v^T+2z_v^T}]\leq e^{C^2(0)vD_2^2}$. This implies
\be
\sup_{v\leq t} \|Y_v^T\|_2\leq e^{\frac{1}{2}C^2(0)tD_2^2}.
\ee
Similarly, $(e^{-j_t^\infty-k_t^\infty})^2=e^{-2j_t^\infty-4k_t^\infty+2k_t^\infty}$. The process $e^{-2j_t^\infty-4k_t^\infty}$ is a supermartingale, and $k_t^\infty\leq\frac{1}{2}C^2(0)tD_2^2$ (by Lemma \ref{est_sigma}). Thus, $$\|e^{-j_t^\infty-k_t^\infty}\|_2\leq e^{C^2(0)tD_2^2}.$$
\end{proof}
\begin{lemma}
\label{L1_HJM_lemma3}
\be
\lim_{T\rightarrow\infty} \|e^{-(j_t^T-j_t^\infty)-(k_t^T-k_t^\infty)}-1\|_2=0.
\eel{L2_jk}
\end{lemma}
\begin{proof}
We need the following two intermediate lemmas.
\begin{lemma}
\label{est_k}
For $T\geq t$, $\sup_{v\leq t}|k_v^T-k_v^\infty|\leq  C(0)C(T-t)tD_2^2$.
\end{lemma}
\begin{proof}
\be
\begin{array}{ll}
\displaystyle{\sup_{v\leq t}|k_v^T-k_v^\infty|} &=\displaystyle{\sup_{v\leq t}\left|\frac{1}{2}\sum_{j\in\mathbb{N}}\int_0^v \left(\big(\int_0^{T-s}+\int_0^\infty\big) \sigma_s^j(u)du\right)\left(\int_{T-s}^\infty \sigma_s^j(u)du\right)ds\right|}\\
&\leq \displaystyle{\sum_{j\in\mathbb{N}}\int_0^{t} \left(\int_0^\infty |\sigma_s^j(u)|du\right)\left(\int_{T-s}^\infty |\sigma_s^j(u)|du\right) ds}\\
&\leq \displaystyle{\sum_{j\in\mathbb{N}} \int_0^{t} C(0) \|\sigma_s^j\|_{\bar{w}} C(T-s)\|\sigma_s^j\|_{\bar{w}} ds}\\
&\displaystyle{=C(0)C(T-t)\int_{0}^{t}\sum_{j\in\mathbb{N}}  \|\sigma_s^j\|_{\bar{w}}^2 ds}\\
&=\displaystyle{C(0)C(T-t)\int_0^{t} \|\sigma_s\|_{L_2^0(H_{\bar{w}})}^2 ds}\\
&\leq C(0)C(T-t)t D_2^2.\\
\end{array}
\ee
\end{proof}
\begin{lemma}
\label{L1_HJM_lemma2}
\be
\lim_{T\rightarrow\infty} \|Y_t^T-1\|_2=0.
\eel{L2_Y}
\end{lemma}
\begin{proof}
By It\^{o}'s formula,
\be
Y_t^T=1+ \int_0^t Y_v^T\bar{\sigma}_v^T\cdot dW^{\mathbb{Q}}_v.
\ee
By It\^{o}'s isometry, we have
\be
\|Y_t^T-1\|_2^2=\mathbb{E}^{\mathbb{Q}}\Big(\int_0^t \|Y_v^T \bar{\sigma}_v^T\|_{\ell^2}^2 dv\Big).
\ee
By Lemma \ref{HJM_additional1}, $|\bar{\sigma}_v^{T,j}|\leq C(T-v)\|\sigma_v^j\|_{\bar{w}}$. Thus
\be
\begin{array}{ll}
\displaystyle{\|Y_t^T-1\|_2^2} & \displaystyle{\leq\mathbb{E}^{\mathbb{Q}}\Big(\sum_{j\in\mathbb{N}}\int_0^t |Y_v^T|^2 C^2(T-v)\|\sigma_v^j\|^2_{\bar{w}} dv\Big)}\\
& \displaystyle{\leq C^2(T-t)\mathbb{E}^{\mathbb{Q}}\Big(\int_0^t |Y_v^T|^2 \sum_{j\in\mathbb{N}}\|\sigma_v^j\|^2_{\bar{w}} dv\Big)}\\
&\displaystyle{=C^2(T-t)\mathbb{E}^{\mathbb{Q}}\Big(\int_0^t |Y_v^T|^2 \|\sigma_v\|^2_{L^0_2(H_{\bar{w}})} dv\Big)}\\
&\displaystyle{\leq C^2(T-t)\mathbb{E}^{\mathbb{Q}}\Big(\int_0^t |Y_v^T|^2 D_2^2dv\Big)}\\
&\displaystyle{\leq C^2(T-t)D_2^2\int_0^t \mathbb{E}^{\mathbb{Q}}(|Y_v^T|^2) dv}\\
&\displaystyle{\leq C^2(T-t)D_2^2\int_0^t C^2 dv}\enskip\text{(Lemma \ref{L1_HJM_lemma1})} \\
&=C^2(T-t)D_2^2C^2t.\\
\end{array}
\ee
Since $\lim_{T\rightarrow\infty}C(T-t)=0$, Eq.\eqref{L2_Y} is verified.
\end{proof}
Now we return to the proof of Lemma \ref{L1_HJM_lemma3}.
\be
\begin{array}{ll}
\|e^{-(j_t^T-j_t^\infty)-(k_t^T-k_t^\infty)}-1\|_2 & =\|Y_t^T e^{z_t^T-(k_t^T-k_t^\infty)}-1\|_2\\
 & \leq \|(Y_t^T-1)e^{z_t^T-(k_t^T-k_t^\infty)}\|_2+\|e^{z_t^T-(k_t^T-k_t^\infty)}-1\|_2.\\
\end{array}
\ee
Recall that by Lemma \ref{est_k}, $|k_t^T-k_t^\infty|\leq C(0)C(T-t)tD_2^2$. Using the same approach as Lemma \ref{est_sigma}, we can show that $$|z_t^T|\leq\frac{1}{2}C^2(T-t)tD_2^2.$$ Thus, we have
\be
\begin{array}{ll}
\|e^{-(j_t^T-j_t^\infty)-(k_t^T-k_t^\infty)}-1\|_2 & \leq \|Y_t^T-1\|_2 e^{\frac{1}{2}C^2(T-t)tD_2^2+C(0)C(T-t)tD_2^2}\\
&\quad+e^{\frac{1}{2}C^2(T-t)tD_2^2+C(0)C(T-t)tD_2^2}-1.\\
\end{array}
\ee
Finally Eq.\eqref{L2_jk} is verified using Lemma \ref{L1_HJM_lemma2} and the fact that $\lim_{T\rightarrow\infty} C(T-t)=0$.
\end{proof}

\bibliography{mybib7}

\end{document}